\documentclass[11pt]{article}

\usepackage[a4paper, margin=1in]{geometry}
\usepackage{amsmath,amssymb,amsthm,amsfonts,cite,xcolor,xspace,graphicx,xifthen,bm}
\usepackage{thmtools,thm-restate}
\usepackage{appendix}
\usepackage{enumitem}
\usepackage{bold-extra}
\usepackage[hang]{footmisc}
\usepackage{subcaption}
\setitemize{noitemsep,topsep=3pt,parsep=3pt,partopsep=3pt}

\definecolor{darkgreen}{rgb}{0,0.5,0}
\definecolor{darkblue}{rgb}{0,0,0.8}
\usepackage{hyperref}
\hypersetup{
    unicode=false,          
    colorlinks=true,        
    linkcolor=darkblue,          
    citecolor=darkgreen,        
    filecolor=magenta,      
    urlcolor=cyan,          
    hypertexnames=false
}

\usepackage[nameinlink,capitalize]{cleveref}

\newtheorem{theorem}{Theorem}[section]
\newtheorem{lemma}[theorem]{Lemma}
\newtheorem{corollary}[theorem]{Corollary}
\newtheorem{conjecture}[theorem]{Conjecture}
\newtheorem{definition}[theorem]{Definition}
\theoremstyle{remark}
\newtheorem*{remark}{Remark}

\newcommand{\LOCAL}{\ensuremath{\mathsf{LOCAL}}\xspace}

\newcommand{\calE}{\ensuremath{\mathcal{E}}}

\newcommand{\poly}{\operatorname{\text{{\rm poly}}}}

\newcommand{\logstar}[1]{\log^{*} #1}

\DeclareMathOperator{\polylog}{\poly\log}

\newcommand{\R}{\mathbb{R}}

\newcommand{\pr}{\operatorname{Pr}}
\newcommand{\rep}{\ensuremath{S_{\operatorname{rep}}}}
\newcommand{\nrep}{\ensuremath{T}}

\newcommand{\Prob}[1]{\pr\left[#1\right]}

\newcommand{\complexityclass}[2][]{\ensuremath{\mathsf{#2}\ifthenelse{\isempty{#1}}{}{(#1)}}}

\newcommand{\pslocal}{\complexityclass{P}\text{-}\complexityclass{SLOCAL}\xspace}

\newcommand{\lovasz}{Lov\'{a}sz\xspace}

\newcommand{\incurved}{incurved}

\newcommand{\gdeg}{\ensuremath{d}}
\newcommand{\badev}{\ensuremath{\mathcal E}}
\newcommand{\inc}{\ensuremath{\operatorname{Inc}}}

\usepackage[most]{tcolorbox}
\newenvironment{quotebox}{
\begin{tcolorbox}[width=0.9\textwidth,
                  enhanced,
                  frame hidden,
                  interior hidden,
                  boxsep=0pt,
                  left=10pt,
                  right=10pt,
                  top=4pt,
                  boxrule=1pt,
                  arc=0pt,
                  colback=white,
                  colframe=black
                  ]
\itshape
}{
\end{tcolorbox}
}

\usepackage{framed}
\usepackage[framemethod=tikz]{mdframed}

\begin{document}

\begin{flushleft}

\vspace*{0.8cm}
{\huge\bf A Sharp Threshold Phenomenon for the Distributed Complexity of the Lov\'asz Local Lemma}
\vspace{1.0cm}
\end{flushleft}

\newcommand{\auth}[3]{\textbf{#1}$\,\,\,\cdot\,\,\,$#2$\,\,\,\cdot\,\,\,$#3\par\medskip}

\auth{Sebastian Brandt}{ETH Z\"{u}rich}{brandts@ethz.ch}

\auth{Yannic Maus\footnote{Supported by the European Union’s Horizon 2020 Research And Innovation Programme
under grant agreement no. 755839.}}{Technion}{yannic.maus@cs.technion.ac.il}

\auth{Jara Uitto}{ETH Z\"{u}rich and University of Freiburg}{jara.uitto@inf.ethz.ch}

\vspace{1cm}

\begin{abstract}
The \emph{Lov\'{a}sz Local Lemma (LLL)} says that, given a set of bad events that depend on the values of some random variables and where each event happens with probability at most $p$ and depends on at most $d$ other events, there is an assignment of the variables that avoids all bad events if the \emph{LLL criterion} $ep(d+1)<1$ is satisfied.  Nowadays, in the area of distributed graph algorithms it has also become a powerful framework for developing---mostly randomized---algorithms.
	A classic result by Moser and Tardos yields an $O(\log^2 n)$ algorithm for the distributed  Lov\'{a}sz Local Lemma [JACM'10] if $ep(d + 1) < 1$ is satisfied.
	Given a \emph{stronger} criterion, i.e., demanding a smaller error probability, it is conceivable that we can find better algorithms.
	Indeed, for example Chung, Pettie and Su [PODC'14] gave an $O(\log_{epd^2} n)$ algorithm under the $epd^2 < 1$ criterion.
	Going further, Ghaffari, Harris and Kuhn introduced an $2^{O(\sqrt{\log \log n})}$ time algorithm given $d^8 p = O(1)$ [FOCS'18].
	On the negative side, Brandt et al.\ and Chang et al.\ showed that we cannot go below $\Omega(\log \log n)$ (randomized) [STOC'16] and $\Omega(\log n)$ (deterministic) [FOCS'16], respectively, under the criterion $p\leq 2^{-d} $. 
	Furthermore, there is a lower bound of $\Omega(\log^* n)$ that holds for \emph{any} criterion.
	
	In this paper, we study the dependency of the distributed complexity of the LLL problem on the chosen LLL criterion.
	We show that for the fundamental case of each random variable of the considered LLL instance being associated with an \emph{edge} of the input graph, that is, each random variable influences at most two events, a sharp threshold phenomenon occurs at $p = 2^{-d}$: we provide a simple deterministic (!) algorithm that matches the $\Omega(\log^* n)$ lower bound in bounded degree graphs, if $p < 2^{-d}$, whereas for $p \geq 2^{-d}$, the $\Omega(\log \log n)$ randomized and the $\Omega(\log n)$ deterministic lower bounds hold.
	
	In many applications variables affect more than two events; our main contribution is to extend our algorithm to the case where random variables influence at most three different bad events. 
	We show that, surprisingly, the sharp threshold occurs at the exact same spot, providing evidence for our conjecture that this phenomenon always occurs at $p = 2^{-d}$, independent of the number $r$ of events that are affected by a variable.
	Almost all steps of the proof framework we provide for the case $r=3$ extend directly to the case of arbitrary $r$; consequently, our approach serves as a step towards characterizing the complexity of the LLL under different exponential criteria.
\end{abstract}
\setcounter{page}{0}
\thispagestyle{empty}
\newpage

\section{Introduction}
The probabilistic method is a standard tool for proving the existence of combinatorial objects satisfying a set of properties $\mathcal{P}$. 
It states that if the probability that a randomly chosen object---where objects are chosen from an appropriate class of objects---satisfies all properties in $\mathcal{P}$ is larger than zero then there exists an object that satisfies all properties in $\mathcal{P}$. 
One can characterize the properties in $\mathcal{P}$ through a set of ``bad events'' $\mathcal{E}_1, \ldots, \mathcal{E}_n$ over some random variables.
Then, showing the existence of an object is the same as showing that the probability that no bad event happens is strictly positive, i.e., there is an assignment of the variables such that none of the bad events occur.
Suppose that we can bound the probability $\Prob{\calE_i} < 1/n$ for each event $\calE_i$.
Then, we can use the well-known \emph{union bound} and obtain that $\Prob{\cap\bar{\calE}_i} > 0$.
The downside of the union bound is that it is \emph{global} in the sense that the bound we require gets more demanding with the number of events.
The celebrated \lovasz Local Lemma (LLL)~\cite{Erdoes1974} can be seen as a \emph{local} version of the union bound.
Let $d$ be the maximum number of other events that any event $\mathcal{E}_i$ depends on and let $p$ be an upper bound on the probability that event $\mathcal{E}_i$ occurs.\footnote{In the \emph{symmetric} version of LLL, one assumes that $p$ is the same for all events $\mathcal{E}_i$.}
Then, the lemma states that if the local $LLL$ \emph{criterion} $ep(d + 1) < 1$ is satisfied, then the probability of avoiding all bad events is strictly positive and due to the probabilistic method there is also an assignment to the variables avoiding all bad events.
Notice that unlike in the union-bound approach, the criterion does not depend on the number $n$ of bad events.

In its original form, the proof of the lemma is existential and does not provide an efficient algorithm for finding an assignment of the variables that avoids the bad events.
For 15 years, finding such a method eluded the research community.
After the first constructive proof and an algorithm for the lemma by Beck~\cite{Beck1991}, a lot of work was dedicated to improving the algorithms for LLL~\cite{Alon1991, Molloy1998, Czumaj2000, Srinivasan2008, Pach2009}.
In 2010, Moser and Tardos~\cite{MoserTardos10} provided a beautiful solution to the problem using a simple algorithm that iteratively (re-)samples all the random variables associated with some bad event that occurred.
Their work was particularly good news for the distributed community; the approach is easy to parallelize.

\paragraph{Distributed LLL} 
Let $V = \{\mathcal{E}_1, \ldots, \mathcal{E}_n\}$ be a finite set of (bad) events that depend on a set of random variables with a discrete finite range. 
The \emph{dependency graph} of an LLL instance is a graph where $V$ is the node set and two nodes are connected by an edge if the corresponding events depend on a common variable.
In the \emph{distributed LLL}, the dependency graph corresponds to a communication network and the model of computing is  the standard distributed message passing model on graphs,  the \LOCAL model:
The nodes communicate in synchronous rounds and in every round, each node can send one (unbounded size) message to each of its neighbors and perform (unbounded) local computations.
In the end of the computation, each node has to know the assignment to the variables that influence its own event and agree on the assignment with its neighbors (that share the variables).
A distributed LLL instance is solved if the nodes have jointly computed an assignment to the variables that avoids all bad events.
The complexity measure is the number of rounds.

\paragraph{Criteria vs.\ Time.}
It is natural to assume that if we strengthen the LLL criterion $ep(d + 1) < 1$, i.e., the bad events are less likely to happen, we can find better algorithms.
From previous work we know, that if $p \geq 2^{-\gdeg}$, one cannot hope to obtain a better runtime than $\Omega(\log \log n)$~\cite{LLL_lowerbound} randomized and $\Omega(\log n)$ deterministic~\cite{Chang2016a}.
It is known that $\Omega(\log^* n)$ cannot be beaten under \emph{any} LLL criterion that is a function of $d$~\cite{SuLLL2017}.
Central questions in the field are how close can we get to these bounds and under which criteria.

\paragraph{Our Contributions.}
Our first contribution is a simple and natural deterministic algorithm that solves the LLL problem in time $O(d + \log^* n)$ under the exponential criterion $p<2^{-\gdeg}$. 
Our algorithm works under the assumption that each random variable affects at most two bad events.
For constant degree graphs, our algorithm is optimal in two ways. 
First, it matches the  $\Omega(\log^* n)$ lower bound and second, the $O(\log^* n)$ runtime which we obtain is unachievable for any criterion that is weaker than $p < 2^{-\gdeg}$.
In particular, crossing the threshold $p = 2^{-\gdeg}$ introduces a fundamental phase shift in the complexity landscape of LLL.
For randomized algorithms, the complexity drops from $\Omega(\log \log n)$ to $O(\log^* n)$ and for deterministic, the drop is from $\Omega(\log n)$ to $O(\log^* n)$.


Let us explain our first result in more detail:
Since every variable affects at most two events, one can assume that the variables are located on the edges of the dependency graph. 
Our algorithm iterates over these edges and deterministically fixes the random variables shared by the endpoints of the edge.
The order of the edges can be adversarial.
In particular, the assigned values are never changed once they are fixed.
We show that independent of the execution history, we can always find an assignment for the variables on an edge that increases the (conditional) probabilities of the corresponding bad events to occur by a factor of at most $2$ per edge that we consider. 
Since there are at most $d$ edges incident on any event, at the end of this process, the probability of any bad event is at most $p \cdot 2^{d} < 1$.
As all random variables in the probability space are fixed, there is no randomness left and hence, the probability must be $0$.

\begin{restatable}{theorem}{subteclll}
	Consider an LLL instance with the criterion $p < 2^{-\gdeg}$ and where every random variable affects at most $2$ bad events.
There is a \emph{sequential} and \emph{local} deterministic process that computes an assignment to the variables that avoids all bad events.
\label{thm:subTechnicalLLL}
\end{restatable}
The algorithm of \Cref{thm:subTechnicalLLL} is \emph{local} in the sense that it can be run in parallel with the help of a suitable edge-coloring of the dependency graph.
The edge-coloring can be computed, e.g., with the algorithm from \cite{panconesi-rizzi}, which implies the following corollary.
\begin{corollary}\label{cor:rank2}
	Consider an LLL instance with the criterion $p < 2^{-\gdeg}$ and where every random variable affects at most $2$ bad events.
	There is an $O(d + \log^* n)$ round deterministic distributed algorithm that solves the LLL problem.
	\label{thm: rank2local}
\end{corollary}

\paragraph{Beyond Two Events per Variable.} 
In many distributed problems, the output of a node may affect many of its neighbors.
Hence, it is natural to study the LLL setting where the random variables influence more than two events.
Our core question is:
\begin{quotebox}
If $r$ is an upper bound on the number of events that any variable affects, what is the weakest LLL criterion that allows us to solve the problem deterministically in $O(\poly(\gdeg)+\log^* n)$ time?
\end{quotebox}
Our algorithms for the case of $r = 2$ generalize in a straightforward way.
However, this comes with a cost.
First, every time our algorithm fixes an assignment for a random variable, the probabilities of the affected events to occur may increase by a factor of $r$.
Furthermore, every event may depend on $\binom{d}{r - 1}$ different variables\footnote{In principle, the number of variables could be larger. However, it is straightforward to reformulate the instance in a way that combines variables affecting the same $r$ events.}.
Hence, the criterion we need is $p<r^{-\binom{d}{ r - 1}}$.
The second question we want to answer is whether this criterion is inherent to the case of larger $r$.


Our main contribution is to answer the aforementioned questions for the case that $r$ is $3$. Furthermore, most parts of our proof---but unfortunately not all---generalize to all $r>3$.  
If each variable affects at most three events and $p<2^{-\gdeg}$ we provide a sequential process that iterates over the random variables and assigns values to them deterministically; the assigned values are never
changed once they are fixed.
Just as in the case $r=2$, this process can be run in parallel and as a corollary, we obtain a distributed algorithm with the same asymptotic runtime as for $r=2$.
Surprisingly, we do not need any compromise in the LLL criterion. 
\begin{center}
\begin{minipage}{0.9\textwidth}
\begin{mdframed}[hidealllines=true, backgroundcolor=gray!20]
\emph{There is a sharp threshold phenomenon for the distributed complexity of the Lov\'asz Local Lemma at $p=2^{-\gdeg}$ if variables can affect at most three events. The phase shift is from $\Omega(\log n)$ deterministic and $\Omega(\log \log n)$ randomized time above the threshold to $O(\poly \gdeg +\logstar n)$ strictly below the threshold. The threshold is the same as in the case where each variable affects at most two events.
}
\end{mdframed}
\end{minipage}
\end{center}
Formally we prove the following statements for the case $r=3$.
\begin{restatable}{theorem}{mainLLL}\label{thm:mainTechnicalLLL}
	Consider an LLL instance with the criterion $p2^{\gdeg}<1$ and where every random variable affects at most $3$ bad events.
There is a \emph{sequential} and \emph{local} deterministic process that computes an assignment to the variables that avoids all bad events.
\end{restatable}
Here the process is \emph{local} in the sense that the choice of the value for a random variable only depends on the $1$-hop neighborhood of the variable in the dependency graph. Thus, we can fix variables in parallel as long as we do not simultaneously fix variables that influence the same event. 
We obtain our desired runtime by first $2$-hop coloring\footnote{A proper coloring of a graph is a $2$-hop coloring if nodes in distance at most two do not have the same color.} the graph with $O(d^2)$ colors.
Then, we iterate through the color classes and fix \emph{all} the random variables of the nodes in the current color class.
Using the algorithm by Fraigniaud, Heinrich and Kosowski, the coloring can be found in $\widetilde{O}(d) + \log^* n$ time~\cite{fraigniaud16}.

\begin{restatable}{corollary}{comadi} \label{cor:mainDistr}
	Consider an LLL instance with the criterion $p<2^{-\gdeg}$ and where every random variable affects at most $3$ bad events.
	There is an $O(d^2 + \log^* n)$ round deterministic distributed algorithm that solves the LLL problem. 
	\label{cor: localalgo}
\end{restatable}
This deterministic algorithm improves on the previously best randomized algorithm which has a runtime of $\omega(\poly\log\log n)$ \cite{newHypergraphMatching} (see related work section for more details).

In a recent breakthrough, Ghaffari, Harris, and Kuhn gave a general derandomization method for \LOCAL algorithms that requires, per node, a small \emph{global} failure probability, i.e., small as a function of $n$.
Put into this frame, our result is a first step towards a general derandomization method under a \emph{local} bound on the failure probability.
Notice that even the very strong criterion $p < 2^{-d}$ is much less demanding than $p < 1/n$ in the case of low degree graphs.
Generalizing it to weaker criteria and random variables that affect arbitrarily many events would be a major breakthrough.

\paragraph{Applications}
In the sinkless orientation problem, the goal is to assign an orientation to each edge such that no node forms a sink.
It is the prime example of a problem that is just very slightly above the "exponential threshold" for which our derandomization works. In fact the failure probability for each node is exactly at the threshold $2^{-d}$ if each edge is oriented randomly and the problem also served for proving the randomized $\Omega(\log\log n)$ lower bound \cite{LLL_lowerbound} and deterministic $\Omega(\log n)$ \cite{Chang2016a} lower bound for LLL in the regime $p\geq 2^{-d}$. It has been studied vividly since its first appearance, see e.g.~\cite{LLL_lowerbound, Chang2016a}, and has also become an important subroutine for solving classic problems such as edge coloring~\cite{ghaffari17, Splitting17}, and hence, studying its relaxations is a canonical step for future work. A natural extension to this problem is to consider orientations in \emph{hypergraphs}, where each edge is allowed to contain more than two nodes. A hyperedge is assigned an orientation towards one of its nodes, i.e., exactly one node is chosen as the head of the hyperedge and the rest are chosen as tails.
For example, the following setting falls into our regime:
Consider hypergraphs of rank 3, i.e., the maximum number of nodes per edge is 3. 
Node $v$ is a sink if it is the head of all of its hyperedges.
The goal is to compute 3 orientations of the edges such that each node $v$ is not a sink in at least two of the orientations. For the parameters to work out, the degree of the dependency graph must be at least 7.

Another problem that is slightly above the exponential threshold is the weak splitting problem~\cite{SLOCAL17} which is defined as follows. You are given a bipartite graph $B=(V\cup U, E)$ with the objective to color the nodes in $U$ with two colors such that each node in $V$ has at least one neighbor in $U$ of each color. 
The seemingly simple problem has recently gotten a lot of attention because it is \pslocal-complete, that is, an efficient, i.e., $\polylog n$ rounds, deterministic algorithm for the problem would imply an efficient deterministic algorithm for many classic problems such at MIS and $\Delta+1$-vertex coloring. Interpreting the vertices in $U$ as variables---the maximum degree of nodes in $U$ corresponds to our parameter $r$, i.e., how many events share a variable--- 
makes it a canonical problem that fits the LLL framework with very weak guarantees. However, the standard case of weak splitting with $2$ colors and the requirement that every node sees at least $1$ neighbor of each color is slightly above the "exponential threshold". Hence, it does not fall into the regime that we can solve with the results in this work. As a matter of fact, there is a $\Omega(\log n)$ lower bound for the problem \cite{weakSplitting19} which is proven through a reduction from the sinkless orientation problem. 

By  weakening the conditions of the weak splitting problem a tiny bit, the problem also fits into our framework and we immediately acquire new algorithms. As an example, consider the weak splitting with $r\leq 3$, where $r$ is the maximum degree of $U$, using $16$ colors and requiring every node to see at least $2$ colors. Notice that a natural interpretation of this setting is a hypergraph edge-coloring, where the nodes of $U$ correspond to rank $r$ hyperedges.
These type of weak variants have also been studied in the aforementioned submission \cite{weakSplitting19}, however, with the objective to show that they are \pslocal-complete when $r$ can be of arbitrary size. Even though the runtime of our algorithms is polynomial in the maximum degree and thus, mostly efficient for very small degrees, we believe that derandomizing weak splitting variants might pave the way to new ideas, also for solving  the original problem. Variants of the weak splitting have also appeared under the name of frugal edge coloring, e.g., in \cite[Definition 2.5]{HarrisHypergraphMatching18}.

\paragraph{Techniques and Generalization}
Now, we discuss which parts of our approach generalize to the case where $r>3$, and which do not.

In the case of $r=3$, we show that, informally speaking, there is a way to fix the random variables such that no node is unlucky too often.
The rough intuition is that if the probability of a particular bad event $\badev_v$ to occur has increased a lot, it must have neighboring bad events whose probability is (relatively) low.
Hence, when the next random variable affecting these events is to be fixed, there is a way to make sure that the probability of $\badev_v$ does not increase much (or even decreases) since incurring a large increase for the probabilities of the neighboring events is acceptable.
We show this by careful bookkeeping during the execution of our algorithm.
The main technical challenge is to show that each time we fix a random variable, there is a good choice, i.e., a choice that makes sure that certain values we keep track of never grow too much.
In the end, these values bound the increase of the probabilities of the individual bad events, 
The proof is very technical, but ultimately, we reduce existence of a good choice for each random variable to the convexity of a certain function.

For the case of $r = 3$, we were able to express this function as a relatively clean closed expression.
For the case of $r > 3$, finding such an expression and using this knowledge to show that the assoiciated function is convex is the only challenge in obtaining full generality.
All the other parts of our method generalize to higher ranks.
We believe that the approach is viable to higher $r$, but needs new analytic insights that are posed as an exciting direction for future work.


\begin{conjecture}
	Consider an LLL instance with the criterion $p<2^{-\gdeg}$ and where every random variable affects at most $r$ bad events.
	For any $r$, there is a distributed algorithm that solves the LLL problem with criterion $p < 2^{-d}$ in time $O(d^{2} + \log^* n)$.
	\label{conj: highrank}
\end{conjecture}

\begin{remark}
	If we tighten the LLL criterion in \Cref{conj: highrank} to $p < 2^{-\Omega(d^2 \cdot \log d)}$, we can use the deterministic algorithm by Fischer and Ghaffari~\cite{ManuelaLLL17} to solve the LLL in time $O(d^2 + \log^* n)$ in the following way. First, in time $\widetilde{O}(d) + \log^* n$, we compute a $2$-hop coloring of the dependency graph with $O(d^2)$ colors and treat this coloring as a $(O(d^2), 0)$-network decomposition. Then, using this network decomposition with the algorithm by Fischer and Ghaffari, we obtain an LLL algorithm with runtime $O(d^2 + \log^* n)$. 
\end{remark}

\paragraph{Related Work.} 
A straightforward distributed implementation of the resampling framework of Moser and Tardos~\cite{MoserTardos10} yields an $O(\log^2 n)$ algorithm for the distributed LLL problem under the $ep(d + 1) < 1$ criterion.
Recently, Chung et al.~\cite{SuLLL2017} designed an algorithm with runtime $O(\log n \cdot \log^2 d)$ which was then improved to $O(\log n \cdot \log d)$ by Ghaffari~\cite{GhaffariImproved16}. 
 The runtime landscape changes if we make the LLL criterion stronger, i.e., make the bad events less likely to happen.
Under the $epd^2 < 1$ criterion, Chung et al.~\cite{SuLLL2017} gave an $O\big( \log_{1/epd^2} n\big)$ time algorithm for LLL.
For an even stronger polynomial criterion of $epd^{32} < 1$ and assuming that $d = O(\log^{1.5} \log n)$, Fischer and Ghaffari gave a $2^{O(\sqrt{\log \log n})}$ time algorithm~\cite{ManuelaLLL17}.
The state-of-the-art under the polynomial criteria, is a $ \exp^{(i)} O\big( (\log^{(i + 1)} n)^{0.5}  \big) = 2^{o(\sqrt{\log \log n})} \ll \poly(\log n)$~\cite{newHypergraphMatching}\footnote{Here $\exp^{(i)}$ stands for a power tower of height $i$ and $\log^{(i)}$ stands for a logarithm iterated $i$ times.} round algorithm for any $i$ under the $d^8 p = O(1)$ criterion.
Subsequent to the current paper Rhozo\v{n} and Ghaffari obtained a major breakthrough in the area of local distributed graph algorithms by obtaining a simple, deterministic and efficient algorithm to solve the so called network decomposition problem \cite{Rhozon19}. Among many other implications this yields a $\polylog\log n$ randomized algorithm for the distributed LLL problem under some fixed polynomial criterion of the form $d^{O(1)}p<1$.

 Chang and Pettie~\cite{ChangHierarchy17}  underlined the general importance of the LLL problem for distributed randomized algorithms. They showed that in  constant degree graphs, LLL is ``complete'' for problems solvable in sublogarithmic time in the following sense:
Any algorithm that has a runtime of $o(\log n)$ can be automatically turned into a new algorithm that runs in time $O(T_{\textrm{LLL}})$, where $T_{LLL}$ stands for the randomized complexity of solving the LLL problem under \emph{any} polynomial criterion.
The  authors make the following conjecture.
\begin{conjecture}[\cite{ChangHierarchy17}]
Assume that $d \geq 2$.
There exists a sufficiently large constant $c$ such that the distributed
LLL problem can be solved in $O(\log \log n)$ time on bounded degree graphs, under the symmetric LLL
criterion $p < d^{-c}$.
\end{conjecture}
One aspect of their conjecture is to find a threshold criterion that allows for $O(\log \log n)$ time algorithms.
We show that in the exponential end of the spectrum, this threshold can be achieved and even broken without using any randomization.
This picture highlights another interesting question for future work:
What bounds can we achieve for LLL criteria between exponential and polynomial?


All the related work that we surveyed so far are for randomized algorithms and there are very few deterministic results for the distributed LLL problem. 
A $\lambda \cdot n^{1/\lambda} \cdot 2^{\sqrt{\log n}}$ time algorithm under the criterion $p(ed)^\lambda < 1$ is known~\cite{ManuelaLLL17} and the state-of-the-art runtime of $ \exp^{(i)} O\big( (\log^{(i)} n)^{0.5}  \big)$ is by Ghaffari, Harris and Kuhn~\cite{newHypergraphMatching}.





\section{Warm-up: When Variables Affect at Most Two Events}
\label{sec:edgecase}

In this section, we discuss our results for the setting where each random variable affects at most two events.
In other words, each random variable of the given LLL instance belongs to one edge in the dependency graph.
Note that we can assume that there is \emph{exactly} one variable per edge: the definition of the dependency graph ensures in this setting that there is at least one variable per edge, and if an edge is associated with more than one random variable we can encode these random variables in one new random variable.
We start by proving the following statement.
\subteclll*
\begin{proof}
	Our sequential process does not depend on the order in which we fix the random variables; in fact our algorithm still works if an (even adaptive) adversary chooses the order in which we have to fix the random variables.
	Let $X$ be a random variable chosen by this adversary, let $e = \{u, v\}$ be the edge associated with $X$, let $\badev_u$ and $\badev_v$ be the two bad events associated with $u$ and $v$, and let $X_1, \dots X_z$ be the random variables that are affecting at least one of $\badev_u, \badev_v$ and have already been fixed, say, to values $x_1, \dots, x_z$. 
	Let $p_1, \dots, p_k$ denote the probabilities with which $X$ assumes the $k$ possible values $y_1, \dots, y_k$, respectively.

	We claim that there exists a value $y_i$ for $X$ such that, by fixing $X=y_i$, the increase for the probability that $\badev_u$ occurs and the increase for the probability that $\badev_v$ occurs add up to at  most $2$. Introduce the term $\theta=\bigwedge_{i=1}^z(X_i=x_i)$;
	then, in other words, we want to show that there exists a value $y_i$ for $X$ such that
	\begin{align*}
		& \frac{\pr[\badev_u \mid \theta, X=y_i]}{\pr[\badev_u \mid \theta]} 
		+ \frac{\pr[\badev_v \mid \theta, X=y_i]}{\pr[\badev_v \mid \theta]} \leq 2 \enspace.
	\end{align*}
	Suppose our claim is false.	Then we have
	 \begin{align*}
		2 & = \sum_{i=1}^k p_i \cdot 2 \\
		& < \sum_{i=1}^k \left( p_i \cdot \left( \frac{\pr[\badev_u \mid \theta, X=y_i]}{\pr[\badev_u \mid \theta]} + \frac{\pr[\badev_v \mid \theta, X=y_i]}{\pr[\badev_v \mid \theta]} \right) \right) \\
		& = \frac{\sum_{i=1}^k \left( p_i \cdot \pr[\badev_u \mid \theta, X=y_i] \right)}{\pr[\badev_u \mid \theta]} 
		+ \frac{\sum_{i=1}^k \left( p_i \cdot \pr[\badev_v \mid \theta, X=y_i] \right)}{\pr[\badev_v \mid \theta]} \\
		& = \frac{\pr[\badev_u \mid \theta]}{\pr[\badev_u \mid \theta]} + \frac{\pr[\badev_v \mid \theta]}{\pr[\badev_v \mid \theta]} 
		 = 2 \enspace,
	\end{align*}
	yielding a contradiction.
	Hence, the claim is true; in particular, we can choose a value for our random variable $X$ such that the probability that $\badev_u$ occurs increases by a factor of at most $2$ and the same holds for $\badev_v$.
	Imagine that we choose a value with this property for each random variable.
	Then, after all random variables are fixed, the total increase for the probability of a bad event $\badev_v$ to occur is upper-bounded by $2^{\deg(v)} \leq 2^d$ because each random variable associated with an edge incident to $v$ supplies an increase of at most $2$ and all other random variables do not change the probability of $\badev_v$ to occur.
	Since, in the beginning, $\pr[\badev_v]  < 2^{-d}$, the probability of $\badev$ to occur after fixing all random variables in the described way is strictly smaller than $1$ and must therefore be $0$.
	Hence, none of the bad events occurs.
\end{proof}

\begin{proof}[Proof of \Cref{cor:rank2}]
As the variables in the proof of \Cref{thm:subTechnicalLLL} can be fixed simultaneously if they do not influence the same event, we can parallelize the algorithm with an edge-coloring of the dependency graph---recall that there is one variable on each of the edges of the dependency graph. An $O(d)$-edge-coloring of the dependency graph can be computed in $O(d+\logstar n)$ rounds and then we can iterate through the color classes in $O(d)$ rounds and fix all variables.
\end{proof}

\section{An Algorithm for $r=3$}

\label{sec: upper}
In this section we show our main result, a sequential and local process to fix the variables of an LLL with exponential criterion $p<2^{-d}$ where each variable affects at most three events. 
We use \Cref{sec:outl,sec:triple,sec:mainP} for its proof. 
In \Cref{sec:mainCor}, we show how this leads to an $O(d^2+\logstar n)$ distributed algorithm for the LLL problem.
Before outlining our proof, we introduce some necessary notation.

Let $H = (V,F)$ be the hypergraph defined on the same node set as the dependency graph, where, for each variable in the LLL instance, we have one hyperedge connecting exactly the nodes from $V$ that depend on the variable.
The \emph{rank} of a hypergraph is the cardinality of the largest hyperedge.
In this section, we focus on the case that $H$ has rank at most $3$, that is, any variable influences at most three events. 
If a hyperedge has cardinality $k$, then we call the associated random variable a \emph{rank $k$} variable.

Throughout this section, we will assume that all considered random variables are of rank \emph{exactly} $3$.
This does not lose generality, due to the following observation:
If the random variable $Y$ under consideration is of rank $2$, then we can simply extend $Y$ to a rank $3$ random variable by adding a virtual third affected bad event, where each choice for the value of $Y$ does not change the probability of the new bad event to occur.

\subsection{Proof Outline}
\label{sec:outl}
We will show that it is possible to fix values for the random variables in a completely arbitrary order.
Moreover, the value a random variable is fixed to only depends on the current state of the radius-$1$ neighborhood of the random variable.
While it is not obvious at all that fixing random variables in the local way and arbitrary order described above is possible, we show that this is indeed the case, by designing a property which we call $P^*$ such that 
\begin{enumerate}[label=\arabic*)]
\item for each input graph (including already fixed random variables) that satisfies $P^*$, and each random variable that has not been fixed yet, there is a value for the random variable such that fixing the random variable to this value will preserve $P^*$, and 
\item after all random variables are fixed (by subsequently choosing values as promised by 1) ), the $P^*$ guarantees that none of the bad events occurs.
\end{enumerate}
The most challenging part of our work is to come up with such a property and to prove that it satisfies 1) and 2).

We continue with describing property $P^*$ more detailed. During the process of fixing the random variables, we keep track, for each event, how favorable or unfavorable the decisions made so far were.
This is done in a peculiar way:
On each edge $e=\{u,v\}$ of the dependency graph we maintain two non-negative values $e_u$ and $e_v$ with $e_u+e_v\leq 2$, that is, we have one value for each endpoint.
Initially all of these values are set to $1$.
Apart from the sum being at most $2$, these values measure the favorability and unfavorability of previous decisions in the following way: For each node $v$, the probability of the associated bad event to occur---in the probability space spanned by all still unfixed variables---is always upper-bounded by $p$ times the product of the values written on the incident edges on the side of $v$, that is, the probability of the bad event at $v$ is always upper bounded by $p\cdot \Pi_{v\in e} e_v$. At the start, this  is trivially satisfied as this product equals $1$. We say \emph{property $P^*$} holds at some point during our algorithm if the values on the edges are such that the aforementioned conditions hold. We will soon give some intuition on why we can always find these values for the edges when fixing a variable, but let us first show that 2) holds: At the end of the process, when all variables are fixed, we obtain that this product is at most $2^{\gdeg}$ as there are at most $\gdeg$ incident edges in the dependency graph and each value on an edge can be at most $2$. When all random variables are fixed, then, for each bad event, the probability that it occurs is at most $p \cdot \Pi_{v\in e}e_v\leq p\cdot 2^{\gdeg} < 1$, which implies that none of the bad events occurs and 2) holds.

While this explains why it is desirable to define $P^*$ in the aforementioned way (in particular, why the ``bookkeeping'' is performed on the edges of the dependency graph), it is far from clear why, when fixing a random variable, a value for the variable as promised in 1) exists.

To shed some light on the intuitive reason for this existence, recall the case where variables affect at most two events, and assume that we want to fix variable $X$ on edge $e=\{u,v\}$.
A choice for $X$ is \emph{good} if the induced increases of the probabilities that the bad events $\mathcal{E}_u$ and $\mathcal{E}_v$ occur are representable on the single edge $e$ between $u$ and $v$ in the dependency graph.
For $r=2$, \emph{representable} means that the sum of the increases is at most $2$. In \Cref{sec:edgecase}, we formally show that such a choice always exists, mostly due to  linearity of expectation.
For $r=3$, let $\{u,v,w\}$ share a hyperedge on which we want to fix a variable $Y$.
Again, we want to show that there is a choice for $Y$ such that the increases of the probabilities that the events $\mathcal{E}_u,\mathcal{E}_v$ and $\mathcal{E}_w$ occur can be represented on the three edges between $u,v$ and $w$ in the dependency graph.
However, in contrast to rank $2$, the condition on whether a triple $(a,b,c)$ of increases can be represented is much more delicate and relies on  nonlinear relations between the values $a,b$ and $c$.
Thus, for rank $3$ we cannot deduce the existence of a good choice immediately from the linearity of expectation.

To better understand our approach for rank $3$, let us look at the case of rank $2$ from a slightly different point of view, where we do not encode different random variables corresponding to the same edge in one new random variable, but instead allow those random variables to be processed individually. 
%
%
%
%
As explained in \Cref{sec:edgecase} and elaborated upon above, it is always possible to find a value for the random variable corresponding to some edge of the dependency graph such that, for the two endpoints, the increases in the probabilities that the respective bad events occur add up to at most $2$.
As it turns out, an even stronger existential statement of this kind is true: If we have to fix several random variables that all correspond to the same edge, we can do so sequentially, each time ensuring that the so-far obtained \emph{total} increases for the probabilities of the affected two events add up to at most $2$.
In other words, for any two given non-negative values $s,t$ with $s + t \leq 2$, we can find a value for any given random variable in a way that ensures that $s$ times the induced probability increase of the first affected bad event plus $t$ times the induced probability increase of the seconded affected bad event is at most $2$.
This statement can be seen as a weighted version of the original statement from \Cref{sec:edgecase}, and, as we will see, both carry over to the case $r=3$.
We next define  property $P^*$ formally. 
\begin{definition}[Property $P^*$]
Let $G$ be the dependency graph of an LLL instance with bad events $\mathcal{E}_1, \ldots, \mathcal{E}_n$, variables $X_1,\ldots, X_m$ and assume that some of the random variables, say, $X_1, \dots, X_z$, have already been fixed to values $x_1, \dots, x_z$, respectively.
Let $\varphi$ be a function that maps each pair $(e,v) \in E \times V$, where $v$ is an endpoint of $e$, to a value $\varphi_e^v \in [0,2]$.
We say that $(G, \varphi)$ satisfies \emph{property $P^*$} if
\begin{enumerate}
	\item $\varphi_e^v + \varphi_e^u \leq 2$ for all $e = \{ u, v \} \in E$, and
	\item $\pr[\badev_v \mid X_1 = x_1, \dots, X_z = x_z] \leq p \cdot \prod_{e \ni v} \varphi_e^v$ for all $v \in V$.
\end{enumerate}
\end{definition}

The main technical ingredient for proving our main result is the following lemma.

\begin{lemma}[Variable Fixing Lemma]
\label{lem:main}
	Let $X_1, \dots, X_z$ be the random variables that have been fixed so far and $X$ an arbitrary random variable that has not been fixed yet.
	Assume that $(G, \varphi)$ satisfies $P^*$.
	Then there is a value $x$ that $X$ can assume and a function $\psi : \{(e,v) \in E \times V \mid v \in e\} \to [0,2]$ such that $(G, \psi)$ satisfies $P^*$.
	Moreover, if $u$ is a node not contained in the hyperedge in $H$ associated with $X$, then for any edge $e = \{u, w\} \in E$ we have $\psi_e^u=\varphi_e^u $ and $\psi_e^w=\varphi_e^w $.
\end{lemma}
We postpone the proof of \Cref{lem:main} to \Cref{sec:mainP} and first show that using \Cref{lem:main} proving our main theorem is straightforward.
\mainLLL*
\begin{proof}
	In the proof of \Cref{thm:subTechnicalLLL}, we could essentially fix each random variable without using any other information than which values were assigned to the adjacent already fixed random variables.
	In the current proof, during the process of fixing the random variables, we will also maintain a function $\varphi : \{(e,v) \in E \times V \mid v \in e\} \to [0,2]$ such that, at all times, $(G, \varphi)$ satisfies property $P^*$.
	Define $B := \{(e,v) \in E \times V \mid v \in e\}$.
	In the beginning, before fixing any random variable, we set $\varphi_e^v := 1$, for all $(e,v) \in B$, thereby ensuring that $(G, \varphi)$ indeed satisfies property $P^*$.

	Now \Cref{lem:main} asserts that there is a value for any adversarially chosen random variable and an update to our function $\varphi$ such that after assigning this value and updating $\varphi$ accordingly, $(G, \varphi)$ still satisfies property $P^*$.
	Let all random variables be fixed according to this scheme.
	Then, after fixing the random variables, property $P^*$ ensures that the probability of some specific bad event $\badev_v$ to occur is upper-bounded by
	\[
		p \cdot \prod_{e \ni v} \varphi_e^v < 2^{-d} \cdot \prod_{e \ni v} 2 \leq 1 \enspace.
	\]
	Hence, for each bad event, the probability that it occurs is $0$.
\end{proof}

\subsection{Representable Triples}
\label{sec:triple}
The first subproperty of property $P^*$ naturally gives rise to the notion of a very useful class of objects which we call representable triples.
\begin{definition}[Representable triples] \label{def:rep}
A triple $(a,b,c) \in \R_{\geq 0}^3$ is called \emph{representable} if there are values $a_1, a_2, b_1, b_3, c_2, c_3 \in [0,2]$ such that $a_1 \cdot a_2 = a$, $b_1 \cdot b_3 = b$, $c_2 \cdot c_3 = c$, $a_1 + b_1 \leq 2$, $a_2 + c_2 \leq 2$, and $b_3 + c_3 \leq 2$.
Let $\rep=\{(a,b,c)\in \R_{\geq 0}^3 \mid (a,b,c)\text{ is representable}\}$ denote the set of all representable triples.
\end{definition}
\begin{figure}
	\centering
	\includegraphics[width=0.5\textwidth]{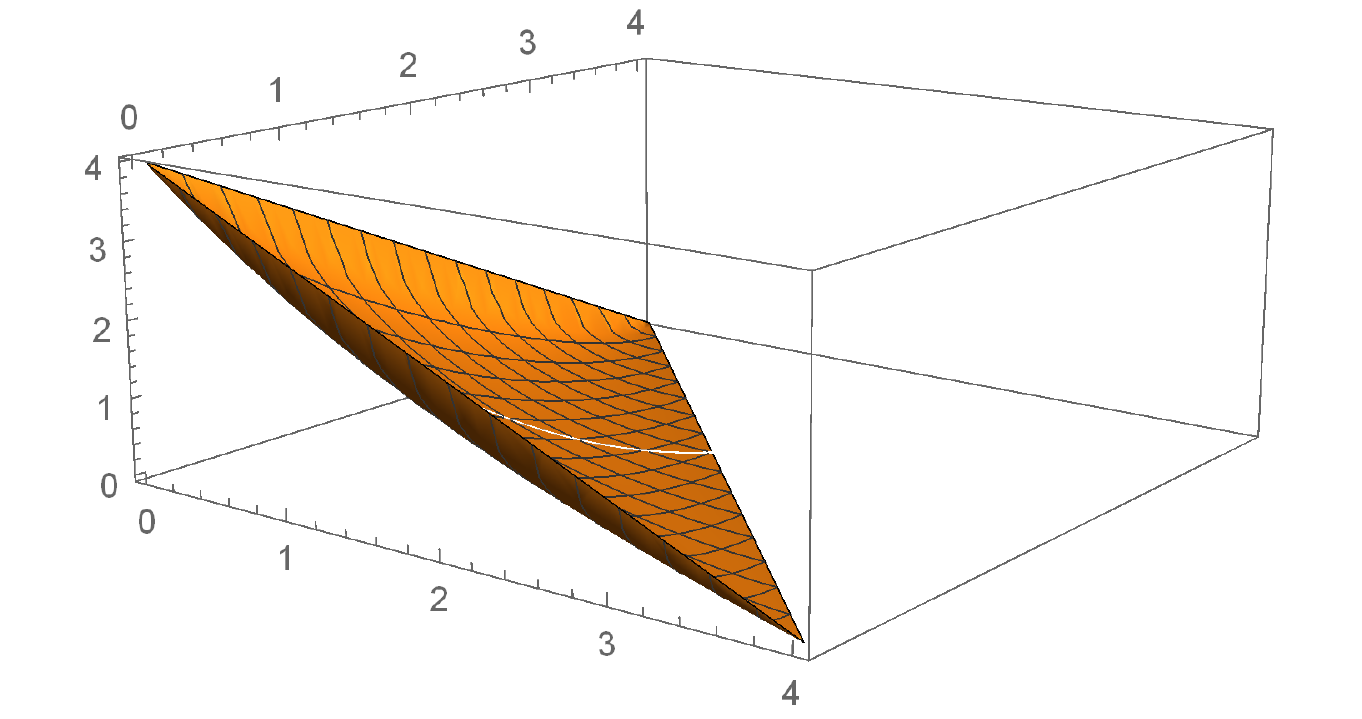}
	\caption{The set ${ \rep}$ of representable triples consists of all points in the first octant that are below the orange shaded surface. The reader might use the plot to convince himself of the fact that the set is incurved, that is, if two points ${s,s'\in\R_{\geq 0}^3}$  are not in ${  \rep}$ then also the line connecting them does not intersect with the set ${\rep}$. }
	\label{fig:Srep}
\end{figure}
See \Cref{fig:reppi} for an illustration of a representable triple and \Cref{fig:Srep} for an illustration of the set $\rep$.
Note that in any representable triple $(a,b,c)$ we have $a,b,c\leq 4$.
\begin{figure}
	\centering
	\includegraphics[width=0.5\textwidth]{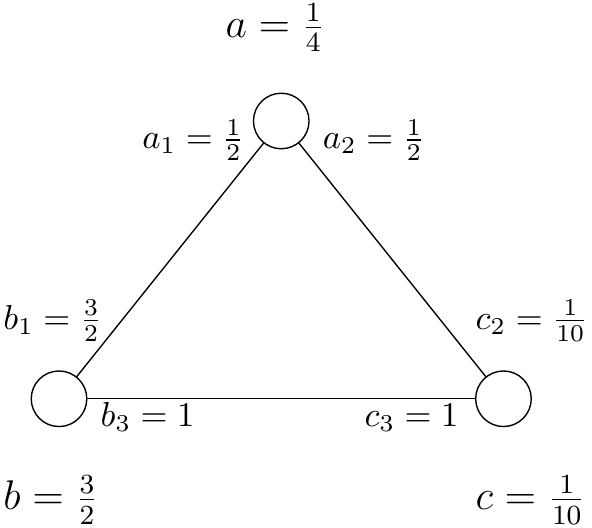}
	\caption{A quick calculation verifies that since $a_1 + b_1 \leq 2$, $a_2 + c_2 \leq 2$, and $b_3 + c_3 \leq 2$, the illustrated triple $(a, b, c) = \left( \frac{1}{4}, \frac{3}{2}, \frac{1}{10} \right)$ is representable.}
	\label{fig:reppi}
\end{figure}
Before we can shed light on how we will use representable triples to obtain the desired result, we need to define another concept that is closely related to convexity.
\begin{definition}[\incurved] We say that a set $S\subseteq \R_{\geq 0}^3$ is \emph{\incurved} if there are no $q \in [0,1]$, $s, s' \in \R_{\geq 0}^3 \setminus S$ such that $q \cdot s + (1-q) \cdot s' \in S$~.
\end{definition}
Now we are set to give a high-level overview of our approach to prove \Cref{lem:main}.
The ``representable triples'' from \Cref{sec:triple} formalize both the situation at the beginning of the variable fixing performed in \Cref{lem:main} and the desired outcome, or more precisely, the part of the situation (resp.\ outcome) corresponding to subproperty (1) of property $P^*$: before fixing the variable in question, the values given by $\varphi$ add up to at most $2$ on each edge, hence the corresponding triple is representable, and \Cref{lem:main} assures that the same holds for the values given by $\psi$ after the fixing.
The situation before the fixing is worst-case if the respective representable triple lies on the surface bounding the set of all representable triples and if we can find a new representable triple (as the desired outcome, together with a suitable value for the random variable), then we can also find one that lies on this surface.
Hence, we can reduce our considerations to this surface.
It turns out that the question whether there exists a suitable value for the random variable to be fixed in \Cref{lem:main} depends on the shape of this surface: if there is no such suitable value, then the function describing the surface is not convex. 
However, \Cref{lem:ess} and \Cref{lem:convex} show that this function is convex, thereby ensuring that a suitable value exists, which essentially proves \Cref{lem:main}. For technical reasons, we use the concept of incurvedness (instead of the closely related convexity) for the proof of the relation between the existence of a suitable value for the random variable and the shape of the surface (\Cref{lem:bad}), and show that the convexity of the surface-describing function implies the incurvedness of the set $\rep$ that is bounded by the function (\Cref{lem:incurved}). 
We encourage the reader to consult the plot given in \Cref{fig:Srep} to convince her- or himself that $\rep$ is incurved.


\begin{lemma}\label{lem:ess}
	\emph{(Proof Deferred to the Appendix)}
	For $f(a,b) := 4 + 1/2 \cdot \big(ab - 2a - 2b - \sqrt{ab(4-a)(4-b)}\big)$ we obtain
	\[\rep = \{ (a, b, c) \subseteq \R_{\geq 0}^3 \mid a + b \leq 4, c \leq f(a,b) \}~.\]
\end{lemma}

\begin{lemma}\label{lem:convex}
	\emph{(Proof deferred to the Appendix)}
	The function $f(a,b) := 4 + 1/2 \cdot (ab - 2a - 2b - \sqrt{ab(4-a)(4-b)})$ is convex on $\{ (a,b) \in \R_{\geq 0}^2 \mid a + b \leq 4 \}$.
\end{lemma}

The characterization of $\rep$ given in \Cref{lem:ess} together with the convexity of $f(a,b)$ from \Cref{lem:convex} imply that $\rep$ is \incurved.
The formal proof of this statement is given by \Cref{lem:incurved}.
Note that \Cref{lem:ess} and \Cref{lem:convex} are not crucial to our general approach if one can find a different proof to show that $\rep$ is \incurved. 
\begin{lemma}\label{lem:incurved}
$\rep$ is \incurved.
\end{lemma}
\begin{proof}
Suppose for a contradiction that there are $q\in [0,1]$, $s, s'\in\R_{\geq 0}^3 \setminus \rep $ with $q\cdot s+(1-q)\cdot s'\in\rep$. Let $s = (a, b, c)$, $s' = (a', b', c')$ and set $(a_q, b_q, c_q) := q \cdot s + (1-q) \cdot s'$.
We consider three cases, depending on the values of $a, b$ and $a', b'$.
\begin{enumerate}
	\item\label{item:leqleq} \underline{$a + b \leq 4$ and $a' + b' \leq 4$.} Consider the two triples $(a, b, \gamma)$ and $(a', b', \gamma')$, where $\gamma = f(a,b)$ and $\gamma' = f(a',b')$. By Lemma~\ref{lem:ess}, we have $(a, b, \gamma)$, $(a', b', \gamma') \in \rep$, and since $s, s' \notin \rep$, we have $c > \gamma$ and $c' > \gamma'$, which implies that $q \cdot c + (1-q) \cdot c' > q \cdot \gamma + (1-q) \cdot \gamma'$. By Lemma~\ref{lem:convex}, we know that $q \cdot \gamma + (1-q) \cdot \gamma' \geq f(q \cdot a + (1-q) \cdot a', q \cdot b + (1-q) \cdot b')$, hence we have $q \cdot c + (1-q) \cdot c' > f(q \cdot a + (1-q) \cdot a', q \cdot b + (1-q) \cdot b')$. In other words, if we set $(a_q, b_q, c_q) := q \cdot s + (1-q) \cdot s'$, then $c_q > f(a_q, b_q)$, which implies $(a_q, b_q, c_q) \notin \rep$, by Lemma~\ref{lem:ess}, and we obtain a contradiction.
	\item \underline{$a + b > 4$ and $a' + b' > 4$.} In this case, we see that, for any $q \in [0,1]$, we have $a_q + b_q > 4$, which implies $(a_q, b_q, c_q) \notin \rep$, by Lemma~\ref{lem:ess}, and, again, we obtain a contradiction.
	\item \underline{$a + b \leq 4$ and $a' + b' > 4$ or vice versa.} W.l.o.g., we can assume that $a + b \leq 4$ and $a' + b' > 4$. Let $r \in [0,1]$ be the uniquely defined value such that, for the triple $s_r := (a_r, b_r, c_r) := r \cdot s + (1-r) \cdot s'$, we have $a_r + b_r = 4$. Observe that if $q \leq r$, we can write $(a_q, b_q, c_q)$ as $q' \cdot s + (1-q') \cdot s_r$, for some $q' \in [0,1]$, by setting $q' := 1 - (1-q)/(1-r)$ (or $q' := 0$ if $r=1$). But then we obtain a contradiction by applying the argumentation in Case~\ref{item:leqleq} for the triples $s$ and $s_r$. Hence, we know that $q>r$. Similarly to before, we can write $(a_q, b_q, c_q)$ as $q' \cdot s_r + (1-q') \cdot s'$, for some $q' \in [0,1)$. In particular, we have $q' \neq 1$ which together with $a_r + b_r = 4$ and $a' + b' > 4$ implies that $a_q + b_q = q' \cdot (a_r + b_r) + (1-q') \cdot (a' + b') > 4$. By Lemma~\ref{lem:ess}, $(a_q, b_q, c_q) \notin \rep$, yielding a contradiction. \qedhere
\end{enumerate}
\end{proof}

\subsection{Proof of the Variable Fixing Lemma (\Cref{lem:main})}
\label{sec:mainP}
Consider the random variable $X$ we are about to fix in \Cref{lem:main}, and assume that it affects the three events located at the nodes $u, v, w$ which then are connected in the dependency graph by the edges $e = \{u, v\}$, $e' = \{u, w\}$, and $e'' = \{v, w\}$.
Recall that our input graph $G$ together with the function $\varphi$ satisfies $P^*$.
Let the discrete distribution of $X$ be given by the set $\{ y_1, \dots, y_k \}$ of possible values for $X$ and corresponding positive probabilities $p_1, \dots, p_k$ with which those values occur.
For each possible value $y$ of $X$ we are interested in the increases (or decreases) of the probabilities of our bad events due to fixing $X=y$. 
For each value $y$ and each $t \in \{u,v,w\}$, define 
\[
	\inc(t, y) := \frac{\pr[\badev_t \mid X_1 = x_1, \dots, X_z = x_z, X=y]}{\pr[\badev_t \mid X_1 = x_1, \dots, X_z = x_z]} \enspace.
\]
If $\pr[\badev_t \mid X_1 = x_1, \dots, X_z = x_z] = 0$, set $\inc(v, y) = 0$.

Now if we can find a value $y$ for $X$ and values $\psi_e^u, \psi_e^v, \psi_{e'}^u, \psi_{e'}^w$, $\psi_{e''}^v, \psi_{e''}^w \in [0,2]$ with $\psi_e^u + \psi_e^v \leq 2$, $\psi_{e'}^u + \psi_{e'}^w \leq 2$, $\psi_{e''}^v + \psi_{e''}^w \leq 2$, such that
\begin{align*}
	\psi_e^u \psi_{e'}^u &\geq \inc(u, y) \cdot \varphi_e^u \varphi_{e'}^u \enspace, \\
	\psi_e^v \psi_{e''}^v &\geq \inc(v, y) \cdot \varphi_e^v \varphi_{e''}^v \enspace, \textrm{ and} \\
	\psi_{e'}^w \psi_{e''}^w &\geq \inc(w, y) \cdot \varphi_{e'}^w \varphi_{e''}^w \enspace,
\end{align*}
then we can prove Lemma~\ref{lem:main} by identifying $\psi$ and $\varphi$ on $\{(e''',v) \in E \times V \mid v \in e''', e''' \notin \{ e, e', e'' \} \}$.
In order to find such a value $y$, we will need the definition of an evil value.
\begin{definition}[evil value]\label{def:evva}
Let $\rep \subseteq \R_{\geq 0}^3$ denote the set of all representable triples, and fix some $(a,b,c) \in \rep$.
We call a value $y$ of a rank $3$ random variable $X$ \emph{$(a,b,c)$-evil} if there exists no $(a',b',c') \in \rep$ such that
	\begin{align*}
		a' &\geq \inc(u, y) \cdot a \enspace, \\
		b' &\geq \inc(v, y) \cdot b \enspace, \textrm{ and} \\
		c' &\geq \inc(w, y) \cdot c \enspace.
	\end{align*}
\end{definition}	
\noindent In particular, if $y$ is $(a,b,c)$-evil, then the triple $(\inc(u, y) \cdot a, \inc(v, y) \cdot b, \inc(w, y) \cdot c)$ is not contained in $\rep$.
%
%
We prove the following.
\begin{lemma}\label{lem:bad}
If there exist a rank $3$ random variable $X$ and a triple $(a,b,c) \in \rep$ such that all possible values $y$ for $X$ are $(a,b,c)$-evil, then $\rep$ is not \incurved.
\end{lemma}
\begin{proof}
	Let $X$ and $(a,b,c)$ be as described above.
	Let $\{y_1, \dots, y_k\}$ be the set of possible values for $X$, occurring with positive probabilities $p_1, \dots, p_k$.
	Let $\badev_u, \badev_v, \badev_w$ denote the bad events affected by $X$, and let $u, v, w$ denote the associated nodes in the dependency graph, respectively.
	For $x \in \{ u, v, w \}$, we denote by $\pr[\badev_x]$ the probability of bad event $\badev_x$ to occur conditioned on the values the already fixed random variables have been assigned.
	By the definition of $\inc(\cdot, \cdot)$, we have, for any $x \in \{ u, v, w \}$,
	\[
		\sum_{i = 1}^k \left( p_i \cdot \inc(x, y_i) \right) = \sum_{i = 1}^k \frac{p_i \cdot \pr[\badev_x \mid X = y_i]}{\pr[\badev_x]} = \frac{\pr[\badev_x]}{\pr[\badev_x]} = 1 \enspace.
	\]
	For $1 \leq i \leq k$, consider the triples $s_i \in \R_{\geq 0}^3$ defined by $s_i := (\inc(u, y_i) \cdot a, \inc(v, y_i) \cdot b, \inc(w, y_i) \cdot c)$.
	Since $y_i$ is $(a,b,c)$-evil for each $1 \leq i \leq k$, we see that no $s_i$ is contained in $\rep$.
	Moreover, by our above observations, we have
	\begin{align*}
		\sum_{i = 1}^k \left( p_i \cdot s_i \right) & = \Bigg( \sum_{i = 1}^k \left( p_i \cdot \inc(u, y_i) \cdot a \right), \sum_{i = 1}^k \left( p_i \cdot \inc(v, y_i) \cdot b \right), \\
													& \phantom{++++} \sum_{i = 1}^k \left( p_i \cdot \inc(w, y_i) \cdot c \right) \Bigg) = (a,b,c) \enspace.
	\end{align*}
	Hence, $\sum_{i = 1}^k \left( p_i \cdot s_i \right) \in \rep$. Consider the sequence $(t_1, \dots, t_k)$ of triples defined by $t_1 := s_1$ and
	\[
		t_j := \frac{\sum_{i=1}^{j-1} p_i}{\sum_{i=1}^{j} p_i} \cdot t_{j-1} + \frac{p_j}{\sum_{i=1}^{j} p_i} \cdot s_j \enspace,
	\]
	for all $2 \leq j \leq k$.
	Let $j^*$ be the smallest index such that $t_{j^*} \in \rep$.
	A straighforward induction shows that
	\[
		t_j =  \sum_{i=1}^j \left( \frac{p_i}{\sum_{i' = 1}^j p_{i'}} \cdot s_i \right) \enspace.
	\]
	Hence, we know that such a $j^*$ exists since $t_k = \sum_{i=1}^k \left( p_i \cdot s_i \right) \in \rep$ as shown above.
	Moreover, since $t_1 = s_1 \notin \rep$, we see that $j^* \geq 2$.
	Thus, due to the choice of $j^*$, by setting $s := t_{j^*-1} \notin \rep$, $s' := s_{j^*} \notin \rep$, and
	\[
		q := \frac{\sum_{i=1}^{j^*-1} p_i}{\sum_{i=1}^{j^*} p_i} \enspace,
	\]
	we obtain the desired $s, s', q$ since $q \cdot s + (1-q) \cdot s' = t_{j^*}\in \rep$.
	Note that all $t_j$ are elements of $\R_{\geq 0}^3$ as they are weighted averages of elements in $\R_{\geq 0}^3$, by the characterization of the $t_j$ given above.
\end{proof}


\begin{proof}[Proof of \Cref{lem:main}]
Recall our setting and observation immediately before \Cref{def:evva}.
By combining \Cref{lem:incurved} and \Cref{lem:bad}, we see that for the random variable $X$ and the representable triple $(a, b, c) := (\varphi_e^u \varphi_{e'}^u, \varphi_e^v \varphi_{e''}^v, \varphi_{e'}^w \varphi_{e''}^w)$, there exist a value $y$ for $X$ and a representable triple $(a', b', c')$ such that
	\begin{align*}
		a' &\geq \inc(u, y) \cdot a \enspace, ~ b' \geq \inc(v, y) \cdot b \enspace, \textrm{ and}\\
		c' &\geq \inc(w, y) \cdot c \enspace.
	\end{align*}
Since $(a', b', c')$ is representable, there exist values $\psi_e^u, \psi_e^v, \psi_{e'}^u, \psi_{e'}^w,$ $\psi_{e''}^v, \psi_{e''}^w \in [0,2]$ with $\psi_e^u + \psi_e^v \leq 2$, $\psi_{e'}^u + \psi_{e'}^w \leq 2$, $\psi_{e''}^v + \psi_{e''}^w \leq 2$ and $a' = \psi_e^u \psi_{e'}^u$, $b' = \psi_e^v \psi_{e''}^v$, $c' = \psi_{e'}^w \psi_{e''}^w$.
Hence,
\begin{align*}
	\psi_e^u \psi_{e'}^u &\geq \inc(u, y) \cdot \varphi_e^u \varphi_{e'}^u \enspace, \\
	\psi_e^v \psi_{e''}^v &\geq \inc(v, y) \cdot \varphi_e^v \varphi_{e''}^v \enspace, \textrm{ and} \\
	\psi_{e'}^w \psi_{e''}^w &\geq \inc(w, y) \cdot \varphi_{e'}^w \varphi_{e''}^w \enspace,
\end{align*}
which implies that, for the function $\psi$ completed by identifying $\psi$ and $\varphi$ on $\{(e''',v) \in E \times V \mid v \in e''', e''' \notin \{ e, e', e'' \} \}$, the pair $(G, \psi)$ satisfies property $P^*$.
Lemma~\ref{lem:main} follows.
\end{proof}

\subsection{Proof of the Main Corollary}
\label{sec:mainCor}
\comadi*
\begin{proof}
We begin by finding a $2$-hop vertex-coloring of the dependency graph with $O(d^2)$ colors.
This can be done in $\widetilde{O}(d + \log^* n)$ time~\cite{fraigniaud16}.
Then, we iterate through the color classes and every node of a color class fixes, one by one but in a single communication round, \emph{all} of its variables that are not fixed yet.
We make two observations.
First, since we are iterating through a $2$-hop-coloring, no two nodes ever fix variables that share an event.
Put otherwise, variables incident on events within at least $3$ hops from each other cannot share an event.
Second, we are effectively fixing the random variables according to \emph{some} order on the variables.
Since \Cref{thm:mainTechnicalLLL} works for \emph{any} order and due to the local nature of the fixing and bookkeeping in the proof of \Cref{thm:mainTechnicalLLL}, our process correctly solves the LLL instance.
Putting the above together, we obtain an $O(d^2 + \log^* n)$ time algorithm for the LLL problem. \qedhere



\end{proof}






\section*{Acknowledgements}
We  thank Mohsen Ghaffari for bringing the considered problem to our attention, together with providing the solution for rank $2$.

%
\bibliographystyle{alpha}
\bibliography{../references}
\clearpage

\appendix
\section{Deferred Proofs}
\begin{proof}[Proof of \Cref{lem:ess}]
Let $\nrep :=\{ (a, b, c) \subseteq \R_{\geq 0}^3 \mid a + b \leq 4, c \leq f(a,b) \}$
We first show that any representable $(a,b,c)$ is in $\nrep$. If $(a,b,c)$ is representable there are $a_1,a_2,b_1,b_3,c_2,c_3\in [0,2]$ with $a_1+b_1\leq 2, a_2+c_2\leq 2, b_3+c_3\leq 2, a=a_1a_2, b=b_1b_3$ and $c=c_2c_3$. 
Thus we obtain $a+b=a_1a_2+b_1b_3\leq 2a_1+2b_1\leq 2a_1+2(2-a_1)=4$~. We now show that the maximal $c$ for which $(a,b,c)$ is representable for given $a,b\in [0,4]$ with $a+b\leq 4$ equals $f(a,b)$.  Note that this implies that any $0\leq c\leq f(a,b)$ is also representable.
We consider several cases for $a$ and $b$ where the only non trivial case is $a,b\neq 0$. 
\begin{itemize}
\item Case $a=b=0$: We can choose $a_1=a_2=b_2=b_3=0$ and $c_2=c_3=2$, i.e, $(a,b,4)$ is representable and we also have $f(0,0)=4$.
\item Case $a=0,b\neq 0$: We can choose $a_1=a_2=0$ and $b_1=c_2=2$. Then we obtain $c\leq c_2c_3=2c_3\leq 2(2-b_3)=2(2-b/2)=4-b$ and we also have $f(0,b)=4-b$.
\item Case $a\neq0, b=0$: The proof of this case is analogous to the case $a=0, b\neq 0$.
\item Case $a,b\neq 0$. First note that $a,b\neq 0$ implies $a,b\neq 4$ because $a+b\leq 4$. For fixed $a,b\notin \{0,4\}$ we vary the value of $a_1$ to see which values of $c$ can be represented. We denote this varying value of $a_1$ by $x$. 
As $2\geq a_2=a/x$ we obtain $x\geq a/2$ and due to $2\geq b_3=b/(2-x)$ we obtain $x\geq 2-b/2$. 
 Thus we obtain that $c\leq c_2c_3\leq (2-a_2)(2-b_3)\leq \left(2-\frac{a}{x}\right)\left(2-\frac{b}{2-x}\right)=:c(x)$ and $a/2\leq x\leq 2-b/2$~.
To find the maximal value of $c$ we compute the derivative $\frac{d}{dx} c(x)$ as
\begin{align*}
 \frac{d}{dx} c(x)=\frac{2 ((a-b)x^2- a(4-b)x + a(4-b))}{(x-2)^2 x^2}~.
\end{align*}
Thus we have $\frac{d}{dx} c(x)=0$ if and only if $(a-b)x^2+ (a b- 4 a )x + 4 a - a b=0$.

\textit{Case $a=b\neq 0$:} The equality is satisfied for $x=1$, that is,  $c\leq c(1)=(2-a)^2=f(a,a)$~.

\textit{Case $a\neq b, a,b\neq 0$:} Let $p=\frac{a(4-b)}{(a-b)}$. Then we get $\frac{d}{dx} c(x)=0$ if and only if $x^2-px+p=0$ which is satisfied for 
\begin{align*}
x_{1,2}=  \frac{p}{2}\pm\sqrt{\frac{p^2-4p}{4}}=\frac{a(4-b)\pm \sqrt{ab(4-a)(4-b)}}{2(a-b)}~.
\end{align*}
The value $x_2=\frac{p}{2}+\sqrt{\frac{p^2-4p}{4}}$ is outside of the range $[a/2,2-b/2]$ as for $b>a$ it is negative and for $a>b>0$ it is larger than $2-b/2$. 
Plugging  $x_1=\frac{p}{2}-\sqrt{\frac{p^2-4p}{4}}$ into $c(\cdot)$ yields
\begin{align*}
c(x_1) & =\left(2-\frac{a}{x_1}\right)\left(2-\frac{b}{2-x_1}\right) \\
& =\left(2-\frac{2 a (a-b)}{a (4 - b) - \sqrt{ab(4 -a)(4- b)}}\right) \\ 
& \phantom{++++} \cdot \left(2-\frac{2 b (a-b)}{-(4-a) b + \sqrt{ab(4-a)(4 - b)}}\right) \\
& =4\cdot \frac{a(4-a)-\sqrt{ab(4 -a)(4- b)}}{a(4-b)-\sqrt{ab(4 -a)(4- b)}} \\
& \phantom{++++} \cdot \frac{b(4-b)-\sqrt{ab(4 -a)(4- b)}}{b(4-a)-\sqrt{ab(4 -a)(4- b)}}
\intertext{Factoring $\sqrt{ab(4 -a)(4- b)}\neq 0$ out of the four terms yields}
& 4\cdot \frac{\frac{\sqrt{a(4-a)}}{\sqrt{b(4-b)}}-1}{\frac{\sqrt{a(4-b)}}{\sqrt{b(4-a)}}-1} \cdot \frac{\frac{\sqrt{b(4-b)}}{\sqrt{a(4-a)}}-1}{\frac{\sqrt{b(4-a)}}{\sqrt{a(4-b)}}-1} \\
 & = 4\cdot\frac{\frac{\sqrt{a(4-a)}-\sqrt{b(4-b)}}{\sqrt{b(4-b)}}}{\frac{\sqrt{a(4-b)}-\sqrt{b(4-a)}}{\sqrt{b(4-a)}}}\cdot \frac{\frac{\sqrt{b(4-b)}-\sqrt{a(4-a)}}{\sqrt{a(4-a)}}}{\frac{\sqrt{b(4-a)}-\sqrt{a(4-b)}}{\sqrt{a(4-b)}}}\\
 & =4\cdot \left(\frac{\sqrt{a(4-a)}-\sqrt{b(4-b)}}{\sqrt{a(4-b)}-\sqrt{b(4-a)}}\right)^2
\intertext{and multiplying nominator and denominator with $\sqrt{a(4-b)}+\sqrt{b(4-a)}\neq 0$ yields}
& 4\cdot \left(\frac{(a-b)\sqrt{(4-a)(4-b)}-(a-b)\sqrt{ab}}{4(a-b)}\right)^2 \\
  & = \left(\frac{\sqrt{(4-a)(4-b)}-\sqrt{ab}}{2}\right)^2 \\
& =4 + \frac{1}{2} \cdot (ab - 2a - 2b - \sqrt{ab(4-a)(4-b)}) \ .
\end{align*}
As $x\in[a/2,2-b/2]$ and we have $c(a/2)=c(2-b/2)=0$ the maximum possible value for $c$ for a given $a,b\in [0,4]$ with $a,b\neq 0$ and $a+b\leq 4$ is  $c(x_1)=f(a,b)$. This concludes the case $a,b\neq 0$. 
\end{itemize}
Summarizing all cases we have shown that any representable $(a,b,c)$ is contained in the set $\nrep$. For the converse direction let $(a,b,c)\in \nrep$, that is, $a,b,c\in [0,4]$, $a+b\leq 4$ and $c\leq f(a,b)$. The previous proof has shown that $(a,b,f(a,b))$ is representable by letting $a_1=x_1, a_2=a/a_1, b_1=2-a_1, b_3=b/b_2, c_2=2-a_2$ and $c_3=2-b_3$. As $c\leq f(a,b)$ we can also represent $(a,b,c)$ by decreasing $c_2$ (or $c_3$) without violating the constraints.
\end{proof}

\begin{proof}[Proof of \Cref{lem:convex}]
	Set $U := \{ (a,b) \in \R_{\geq 0}^2 \mid a + b \leq 4 \}$ and $U' := \{ (a,b) \in \R_{> 0}^2 \mid a + b < 4 \}$.
	Since $f(a,b)$ is continuous on $U$, it is sufficient to show convexity of $f(a,b)$ on the open domain $U'$, which in turn can be showed by proving that at every point $x\in U'$, the Hessian $\nabla^2 f(x)$ of $f$ is positive semi-definite (cf., e.g., \cite[Section 3.1.4]{boyd_vandenberghe_2004}).
	We will prove the slightly stronger statement that $\nabla^2 f(x)$ is positive definite, which, by Sylvester's criterion, is equivalent to the statement that all leading principal minors of $\nabla^2 f(x)$ are positive.
	In other words, we will prove our lemma by showing that
	$\frac{\partial^2 f(x)}{\partial a^2} >0$ and
$		\frac{\partial^2 f(x)}{\partial a^2} \cdot \frac{\partial^2 f(x)}{\partial b^2} - \frac{\partial^2 f(x)}{\partial a \partial b} \cdot \frac{\partial^2 f(x)}{\partial a \partial b}> 0$
	for all $x \in U'$.
	To this end, we first calculate the four involved terms for $x=(a,b) \in U'$.
	We have
	\begin{align*}
		\frac{\partial f(a,b)}{\partial a} & = \frac{1}{2} \cdot \left( b - 2 - \frac{b(4-b)(4-2a)}{2 \sqrt{ab(4-a)(4-b)}} \right) \enspace,
	\end{align*}
	which implies
	\begin{align*}
 		\frac{\partial^2 f(a,b)}{\partial a^2} & = \frac{1}{2} \cdot - \frac{ -4b(4-b)\sqrt{ab(4-a)(4-b)}}{4 ab(4-a)(4-b)} \\
		& \phantom{++++} - \frac{1}{2} \cdot \frac{\frac{2b(4-b)(4-2a)b(4-b)(4-2a)}{2\sqrt{ab(4-a)(4-b)}}}{4 ab(4-a)(4-b)} \\
		& = \frac{1}{2} \cdot \left( \sqrt{\frac{b(4-b)}{a(4-a)}} + \frac{(2-a)^2}{a(4-a)} \cdot \sqrt{\frac{b(4-b)}{a(4-a)}} \right) \\
		& = \frac{2}{a(4-a)} \cdot \sqrt{\frac{b(4-b)}{a(4-a)}}
	\end{align*}
	\begin{align*}
		\frac{\partial f(a,b)}{\partial a \partial b} & = \frac{1}{2} \cdot \Bigg( 1 - \frac{ 2(4-2b)(4-2a)\sqrt{ab(4-a)(4-b)}}{4 ab(4-a)(4-b)} \\
		& \phantom{++++} - \frac{\frac{2b(4-b)(4-2a)a(4-a)(4-2b)}{2\sqrt{ab(4-a)(4-b)}}}{4 ab(4-a)(4-b)} \Bigg) \\
		& = \frac{1}{2} \cdot \left( 1 - \frac{2(2-b)(2-a)}{\sqrt{ab(4-a)(4-b)}} + \frac{(2-a)(2-b)}{\sqrt{ab(4-a)(4-b)}} \right) \\
		& = \frac{1}{2} - \frac{(2-a)(2-b)}{2\sqrt{ab(4-a)(4-b)}} \enspace.
	\end{align*}
	Since our function $f(a,b)$ is symmetric in $a$ and $b$, we also have
	\begin{align*}
		\frac{\partial^2 f(a,b)}{\partial b^2} & = \frac{2}{b(4-b)} \cdot \sqrt{\frac{a(4-a)}{b(4-b)}}\text{ and} \\
		\frac{\partial f(a,b)}{\partial b \partial a} & = \frac{1}{2} - \frac{(2-a)(2-b)}{2\sqrt{ab(4-a)(4-b)}} \enspace.
	\end{align*}
	Hence, we obtain
	\begin{align*}
		\frac{\partial^2 f(a,b)}{\partial a^2} & = \frac{2}{a(4-a)} \cdot \sqrt{\frac{b(4-b)}{a(4-a)}} > 0
	\end{align*}
	\begin{align*}
		& \phantom{=i} \frac{\partial^2 f(a,b)}{\partial a^2} \cdot \frac{\partial^2 f(a,b)}{\partial b^2} - \frac{\partial^2 f(a,b)}{\partial a \partial b} \cdot \frac{\partial^2 f(a,b)}{\partial a \partial b} \\
		& = \frac{2}{a(4-a)} \cdot \sqrt{\frac{b(4-b)}{a(4-a)}} \cdot \frac{2}{b(4-b)} \cdot \sqrt{\frac{a(4-a)}{b(4-b)}} \\
		& \phantom{++++} - \left( \frac{1}{2} - \frac{(2-a)(2-b)}{2\sqrt{ab(4-a)(4-b)}} \right)^2 \\
		& = \frac{4}{ab(4-a)(4-b)} - \frac{1}{4} + \frac{(2-a)(2-b)}{2\sqrt{ab(4-a)(4-b)}} \\
		& \phantom{++++} - \frac{(2-a)^2(2-b)^2}{4ab(4-a)(4-b)} \\
		& = \frac{16 - ab(4-a)(4-b) + 2(2-a)(2-b)\sqrt{ab(4-a)(4-b)}}{4ab(4-a)(4-b)} \\
		& \phantom{++++} - \frac{(a^2-4a+4)(b^2-4b+4)}{4ab(4-a)(4-b)} \\
		& = \frac{-2a^2b^2 + 8a^2b + 8ab^2 - 4a^2 - 4b^2 - 32ab + 16a}{4ab(4-a)(4-b)} + \\
		& \phantom{++++} \frac{16b + (2ab - 4a - 4b + 8)\sqrt{ab(4-a)(4-b)}}{4ab(4-a)(4-b)} \\
		& = \frac{16 - (4 + ab -2a -2b - \sqrt{ab(4-a)(4-b)})^2}{4ab(4-a)(4-b)} \\
		& = \frac{16 - \left(\frac{1}{2} \cdot \left( \sqrt{(4-a)(4-b)} - \sqrt{ab} \right)^2 - 4 \right)^2}{4ab(4-a)(4-b)} 
		> 0 \enspace,
	\end{align*}
	for all $(a,b) \in U'$.
	Here, the last inequality follows from the fact that, for all $(a,b) \in U'$, we have
	$
		0 < \left( \sqrt{(4-a)(4-b)} - \sqrt{ab} \right)^2 < 16
	$
	since $0 < a < 4 - b$, $0 < b < 4 - a$, and $(4-a)(4-b) < 16$.
	As argued above, it follows that $f(a,b) := 4 + 1/2 \cdot (ab - 2a - 2b - \sqrt{ab(4-a)(4-b)})$ is convex on $U$.
\end{proof}

\end{document}